\documentclass[runningheads]{llncs}

\usepackage{graphicx}
\usepackage{amsmath}
\usepackage{amssymb,amsfonts}
\usepackage{algorithm}
\usepackage{algorithmic}
\usepackage{enumerate}
\usepackage{tikz}

\DeclareMathOperator*{\argmax}{arg\,max}

\allowdisplaybreaks

\begin{document}

\newcommand{\mwnwtie}{MWNW$^\text{tie}$}
\newcommand{\mnwtie}{MNW$^\text{tie}$}
\newcommand{\wnw}{\texttt{WNW}}
\newcommand*\circled[1]{\tikz[baseline=(char.base)]{
            \node[shape=circle,draw,inner sep=2pt] (char) {#1};}}

\title{On Maximum Weighted Nash Welfare for \\ Binary Valuations}
%
%
\author{Warut Suksompong\inst{1} \and
Nicholas Teh\inst{2}}
\authorrunning{W. Suksompong and N. Teh}
%
\institute{School of Computing, National University of Singapore, Singapore
\email{warut@comp.nus.edu.sg} \and
Department of Computer Science, University of Oxford, UK\\
\email{nicholas.teh@cs.ox.ac.uk}}
\maketitle              
\begin{abstract}
We consider the problem of fairly allocating indivisible goods to agents with weights representing their entitlements.
A natural rule in this setting is the maximum weighted Nash welfare (MWNW) rule, which selects an allocation maximizing the weighted product of the agents' utilities.
We show that when agents have binary valuations, a specific version of MWNW is resource- and population-monotone, satisfies group-strategyproofness, and can be implemented in polynomial time.

\end{abstract}
%
%
%

\section{Introduction}
\label{sec:intro}

A fundamental problem in economics is how to allocate scarce resources to interested agents in a fair manner, often referred to as \emph{fair division} \cite{brams1996fairdivision,moulin2003fairdivision}.
Applications of fair division include distributing books among public libraries, allotting supplies to districts, and dividing personnel between organizations.

An approach to fairly allocating resources that has received significant attention is the \emph{maximum Nash welfare (MNW)} rule, which chooses an allocation that maximizes the product of the agents' utilities.
MNW is \emph{scale-invariant}---scaling the utility function of an agent by a constant factor does not change the outcome of the rule---and its output satisfies \emph{Pareto-optimality}---no agent can be made better off without making another agent worse off.
Moreover, in the ubiquitous setting where the resource consists of indivisible goods, Caragiannis et al.~\cite{caragiannis2016mnw} showed that an MNW allocation is always \emph{envy-free up to one good (EF1)}, meaning that if an agent envies another agent based on their allocated bundles, then the envy can be eliminated by removing a single good from the latter agent's bundle.
This result holds when the agents have additive valuations over the goods, a common assumption in fair division, which we will also make in this paper.
The combination of EF1 and Pareto-optimality is surprisingly elusive---for example, it is not always satisfied by an allocation that maximizes the utilitarian or egalitarian welfare~\cite{caragiannis2016mnw}.
This ``unreasonable fairness'' of MNW led Caragiannis et al.~to call it the ``ultimate solution'' for the division of indivisible goods under additive valuations.

While most of the fair division literature assumes that the agents participating in the resource allocation exercise have the same entitlement to the resource, this assumption is far from the truth in many situations.
Indeed, the agents may have made different levels of investment in a joint venture and therefore deserve different shares of the resulting assets.
Alternatively, if the agents represent groups like districts or organizations, as in the aforementioned examples, larger groups are clearly entitled to a larger proportion of the resource.
As a result, a recent line of work has investigated the problem of allocating indivisible goods to agents with different \emph{weights} representing their entitlements \cite{babaioff2021fairshare,babaioff2021competitive,chakraborty2020wef,chakraborty2021picking,chakraborty2022weighted,farhadi2019wmms,scarlett2021fairness}.
A natural generalization of MNW to the weighted setting is the \emph{maximum weighted Nash welfare (MWNW)} rule, which selects an allocation maximizing the weighted product of the agents' utilities, where the weights appear in the exponents.
MWNW is Pareto-optimal by definition, and Chakraborty et al.~\cite{chakraborty2020wef} showed that it satisfies a weighted extension of EF1 called \emph{weak weighted envy-freeness up to one good (WWEF1)}.

Despite their attractive properties, MNW and MWNW are not without disadvantages.
First, MNW is not \emph{strategyproof} \cite{halpern2020binary,klaus2002strategy}, which implies that an agent can sometimes benefit by misreporting her true preferences.
Second, the rule violates two highly intuitive properties called \emph{resource-} and \emph{population-monotonicity}~\cite{chakraborty2021picking}: resource-monotonicity says that when an extra good is added, no agent should receive lower utility as a result, while population-monotonicity mandates that if an extra agent joins, no existing agent's utility should increase.
Third, finding or even approximating the maximum Nash welfare is a computationally hard problem \cite{lee2017apx}.
Since MWNW is a generalization of MNW, all of these drawbacks also apply to MWNW.
Recent work has therefore focused on the class of \emph{binary} valuations (sometimes referred to as \emph{binary additive} valuations), wherein the utilities are additive and each agent's utility for each good is either $0$ or $1$.
Under binary valuations, Halpern et al.~\cite{halpern2020binary} proved that a particular form of MNW, which they called \mnwtie{}, is \emph{group-strategyproof} and can be computed in polynomial time.
Group-strategyproofness stipulates that no group of agents can misreport their preferences in such a way that they all benefit.

As Halpern et al.~\cite{halpern2020binary} noted, the case of binary valuations is not merely a theoretical curiosity. 
While such valuations undoubtedly limit the expressiveness of preferences, they allow for very simple preference elicitation.
Indeed, in our earlier examples, books can be classified as being relevant to a particular library or not, supplies are either in surplus or shortage for each district, and organizations can indicate their human resource needs in terms of the personnel that they desire.
Binary valuations also correspond to \textit{approval votes}, which have long been studied in the voting literature \cite{brams2007approvalvoting,kilgour2010approval}.
For these reasons, the binary case has been given special attention in several fair division papers \cite{aleksandrov2015onlinefoodbank,amanatidis2021mnwefx,barman2018greedymnw,bouveret2016conflict,darmann2015nashproduct,freeman2019equitable,halpern2020binary,hosseini2020infowithholding,kyropoulou2019groupallocation}.
In addition, when different entitlements are allowed, binary valuations generalize the well-studied setting of \emph{apportionment}, wherein the goods are all identical (e.g., seats in a parliament) \cite{balinski2001fair,pukelsheim2014proportional}.

\subsection{Our Contributions}

In this paper, we show that MWNW (and therefore MNW) exhibits several desirable properties under binary valuations.
In particular, we consider a specific form of the MWNW rule that we call \mwnwtie{}.
\begin{itemize}
\item First, we show that \mwnwtie{} is resource- and population-monotone under binary valuations. 
Even in the unweighted setting, both results are new to the best of our knowledge.
Since resource-monotonicity corresponds to an important axiom in apportionment called \emph{house-monotonicity},\footnote{See, e.g., \cite[p.~117]{balinski2001fair}. This axiom is also known as \emph{committee-monotonicity} \cite{brill2020approval}, and a violation of it is referred to as the \emph{Alabama paradox}.} this suggests that it may be interesting to consider \mwnwtie{} as an apportionment method.
\item Second, we show that \mwnwtie{} satisfies \emph{group-strategyproofness}, which means that no group of agents can misreport their preferences in a way that all members of the group are better off. 
This generalizes the result of Halpern et al.~\cite{halpern2020binary} from the unweighted setting.
We also consider a stronger notion of group-strategyproofness, whereby it is only required that in the manipulating group, none of the members is worse off and at least one is better off.
We show that even when all agents have the same entitlement, \mwnwtie{} (which reduces to \mnwtie{}) fails to satisfy this stronger version.
\item Third, we present an algorithm that computes an \mwnwtie{} allocation in polynomial time.
This extends the result of Halpern et al.~\cite{halpern2020binary} from the unweighted setting.
\end{itemize}
Together with the known fairness and efficiency guarantees under general additive valuations \cite{chakraborty2020wef}, our results establish MWNW as a strong candidate rule when allocating indivisible goods among agents with binary valuations and arbitrary entitlements.
We also discuss the shortcomings of other candidate rules in this setting as a comparison to MWNW.

\section{Preliminaries}\label{sec:preliminaries}
In the setting of allocating indivisible goods, we are given a set $N = \{1,\dots,n\}$ of $n$ \textit{agents} and a set $G = \{g_1, \dots, g_m \}$ of $m$ \textit{goods}. 
Both $N$ and $G$ are not necessarily fixed, as extra agents or goods may be added.
Subsets of goods in $G$ are referred to as \textit{bundles}. Each agent $i \in N$ has a \textit{weight} $w_i > 0$ (representing her entitlement), and a nonnegative \textit{valuation function} $v_i$ over bundles of goods. 
We assume that $v_i$ is \emph{binary} (or \emph{binary additive}), meaning that $v_i(\{g\}) \in \{0,1\}$ for all $i\in N$ and $g\in G$, and $v_i(S) = \sum_{g \in S} v_i(\{g\})$ for any $S\subseteq G$.
For the sake of simplicity, we sometimes write $v_i(g)$ instead of $v_i(\{g\})$ for $g\in G$.
Let the vector of agent weights be $\mathbf{w} = (w_1, \dots, w_n)$, and the vector of agent valuations (which we call the \textit{valuation profile) }be $\mathbf{v} = (v_1, \dots, v_n)$. A problem \textit{instance} $\mathcal{I}$ is defined by the set of agents, goods, weights, and valuation functions.
 
An \emph{allocation} $\mathcal{A} = (A_1, \dots, A_n)$ is a list of $n$ bundles such that no two bundles overlap, where agent $i$ receives bundle $A_i$; let $\Pi_n(G)$ denote the set of all possible allocations.
For an allocation $\mathcal{A}$, let its \textit{utility vector} be $(v_1(A_1), \dots, v_n(A_n))$ and its \textit{weighted utility vector} be $(v_1(A_1)^{w_1}, \dots, v_n(A_n)^{w_n})$.
For a subset of agents $S\subseteq N$, let  $\mathcal{A}_S$ denote the allocation derived from restricting $\mathcal{A}$ to the bundles of the agents in $S$ (in the same order as in $\mathcal{A})$. 
An \textit{allocation rule} is a function that maps each instance to an allocation. 
We say that a good $g$ is \textit{unvalued} if $v_i(g) = 0$ for all agents $i\in N$, and \textit{valued} otherwise. 
An allocation $\mathcal{A}$ is said to be \textit{minimally complete} if it allocates all valued goods and does not allocate any unvalued good.

We can now state the definition of the MWNW rule.

\begin{definition}[MWNW]
    An allocation $\mathcal{A}$ is a \emph{maximum weighted Nash welfare (MWNW) allocation} if, among the set of allocations in $\Pi_n(G)$, it maximizes the number of agents receiving positive utility and, subject to that, maximizes the weighted product of positive utilities. 
    Formally, let $P(\mathcal{A}) = \{i \in N : v_i(A_i) > 0 \}$ and $\mathcal{P} = \argmax_{\mathcal{A} \in \Pi_n(G)} |P(\mathcal{A})|$. Then, $\mathcal{A}$ is an \emph{MWNW} allocation if $\mathcal{A} \in  \argmax_{\mathcal{A'}\in \mathcal{P}} \prod_{i \in P(\mathcal{A'})} v_i(A'_i)^{w_i}$.
    
    The \emph{MWNW rule} selects an \emph{MWNW} allocation from any instance, breaking ties arbitrarily if there is more than one such allocation. 
\end{definition}

When all weights are equal (sometimes referred to as the \textit{unweighted} setting), the MWNW rule reduces to the well-known maximum Nash welfare (MNW) rule.
Halpern et al.~\cite{halpern2020binary} introduced the rule \mnwtie, a version of MNW with additional tie-breaking specifications.
We will extend \mnwtie{} to the weighted setting.
An allocation $\mathcal{A}$ is said to be \textit{lexicographically dominating} in a set of allocations $\mathcal{C}$ if it maximizes the weighted utility vector in a lexicographical order. 
More precisely, among all allocations in $\mathcal{C}$, the allocation $\mathcal{A}$ maximizes $v_1(A_1)^{w_1}$, then subject to that, it maximizes $v_2(A_2)^{w_2}$, and so on. The main allocation rule that we consider in this paper is the following.

\begin{definition}[\mwnwtie{}]
\label{def:mwnwtie}
    The rule \emph{\mwnwtie{}} returns an allocation $\mathcal{A}$ such that
    \begin{enumerate}
        \item $\mathcal{A}$ is an \emph{MWNW} allocation that is also lexicographically dominating in the set of all \emph{MWNW} allocations; and
        \item $\mathcal{A}$ is minimally complete.
    \end{enumerate}
\end{definition}
If there are several allocations satisfying both conditions, then \mwnwtie{} arbitrarily picks one of them. 
Note that even though there can be more than one \mwnwtie{} allocation, all such allocations have the same (weighted) utility vector.
\mwnwtie{} reduces to \mnwtie{} when all weights are equal. 

We now establish two properties of \mwnwtie{} that will be useful for proving our results.

\begin{lemma}\label{lemma-mwnw-value}
    Under binary valuations, in any \emph{\mwnwtie{}} allocation $\mathcal{A}$, every agent values every good in her own bundle. That is, $\forall i \in N$, $\forall g \in A_i$, $v_i(g) = 1$.
\end{lemma}
\begin{proof}
    Suppose, for a contradiction, that there exists a good $g \in A_i$ such that $v_i(g) = 0$. Since any \mwnwtie{} allocation $\mathcal{A}$ is minimally complete, some agent $j \in N$ must have positive value for it: $v_j(g) = 1$. If the weighted Nash welfare of $\mathcal{A}$ is greater than $0$, we can strictly improve the weighted Nash welfare by assigning $g$ to $j$ instead of $i$. On the other hand, if the weighted Nash welfare of $\mathcal{A}$ is 0, we can increase the number of agents receiving positive utility or increase the product of positive utilities by assigning $g$ to $j$. Both cases contradict the assumption that $\mathcal{A}$ is an \mwnwtie{} allocation.
     \hfill $\square$  
\end{proof}

\begin{lemma}\label{lemma-subset-mwnw}
    Under binary valuations, suppose we have an \emph{\mwnwtie{}} allocation $\mathcal{A} = (A_1, \dots, A_n)$. Then, for any subset of agents $S \subseteq N$, $\mathcal{A}_S$ is also an \emph{\mwnwtie{}} allocation for the agents in $S$.
\end{lemma}
\begin{proof}
    Suppose, for a contradiction, that $\mathcal{A}_S$ is not an \mwnwtie{} allocation. This means that there exists another allocation $\mathcal{A}_S'$ of the goods in $\mathcal{A}_S$ to the agents in $S$ such that
    \begin{enumerate}[(i)]
        \item strictly more agents receive positive utility, or
        \item the same number of agents receive positive utility, but the weighted Nash welfare of these agents is strictly higher, or 
        \item the same number of agents receive positive utility, with the same weighted Nash welfare for these agents, but the weighted utility vector of $\mathcal{A}_S'$ is lexicographically preferred to that of $\mathcal{A}_S$.
    \end{enumerate}
    Case (i) means that the number of agents receiving positive utility in $\mathcal{A}$ is not maximized, a contradiction. In case (ii), if we retain the allocation for agents in $N\setminus S$ as in $\mathcal{A}$ while using the allocation $\mathcal{A}_S'$ for agents in $S$, then in the new allocation, the same number of agents receive positive utility as in $\mathcal{A}$, but the weighted Nash welfare of these agents is strictly higher, a contradiction.
    Finally, in case (iii), this means that $\mathcal{A}$ is not lexicographically dominating, yielding yet another contradiction. \hfill $\square$ 
\end{proof}

To end this section, we introduce the notion of a \textit{transformation graph}. 
Consider two allocations $\mathcal{A}$ and $\mathcal{A}'$, where the set of goods allocated may be different (but the set of agents is the same).
Let the transformation graph $\mathcal{G}(\mathcal{A}, \mathcal{A'})$ from $\mathcal{A}$ to $\mathcal{A'}$ be a directed graph where the nodes represent the agents, and for each good $g$ such that $g\in A_i$ and $g\in A_j'$ for some $i\ne j$, we add an edge $i\rightarrow j$. 
In particular, there may be multiple edges from one node $i$ to another node $j$.

\section{Monotonicity}\label{sec:monotonicity}

In this section, we show that \mwnwtie{} satisfies resource- and population-monotonicity.
As a by-product of our proof for resource-monotonicity, we will also obtain a polynomial-time algorithm that computes an \mwnwtie{} allocation.
We first state the definitions of the two monotonicity properties---both of them have been studied in recent fair division papers \cite{chakraborty2021picking,segal2017monotonicity,segal2019monotonicity}.

\begin{definition}[Resource-monotonicity]
    An allocation rule $f$ is \emph{resource-monotone} if the following holds: For any two instances $\mathcal{I}$ and $\mathcal{I'}$ such that $\mathcal{I'}$ can be obtained from $\mathcal{I}$ by adding one extra good, if $f(\mathcal{I}) = \mathcal{A}$ and $f(\mathcal{I'}) = \mathcal{A'}$, then for each $i \in N$, $v_i(A_i) \leq v_i(A'_i)$.
\end{definition}

\begin{definition}[Population-monotonicity]
    An allocation rule $f$ is \emph{population-monotone} if the following holds: For any two instances $\mathcal{I}$ and $\mathcal{I'}$ such that $\mathcal{I'}$ can be obtained from $\mathcal{I}$ by adding one extra agent, if $f(\mathcal{I}) = \mathcal{A}$ and $f(\mathcal{I'}) = \mathcal{A'}$, then for each agent $i$ in the original set of agents $N$, $v_i(A_i) \geq v_i(A'_i)$.
\end{definition}

We now proceed to our first main result.

\begin{theorem}\label{thm-resourcemonotone}
    Under binary valuations, the \emph{\mwnwtie{}} rule is resource-monotone.
\end{theorem}

\begin{proof}
    Let $\mathcal{I}$ and $\mathcal{I'}$ be instances such that $\mathcal{I'}$ can be obtained from $\mathcal{I}$ by adding one extra good $g^*$. Suppose that \mwnwtie{} returns $\mathcal{A}$ on $\mathcal{I}$ and $\mathcal{A'}$ on $\mathcal{I'}$.\footnote{Recall that for each instance, all \mwnwtie{} allocations induce the same (weighted) utility vector.}
    By Lemma~\ref{lemma-mwnw-value}, in both $\mathcal{A}$ and $\mathcal{A'}$, every agent values every good that she receives.
    Let $\mathcal{A''}$ be the allocation that is equivalent to $\mathcal{A'}$ except that $g^*$ is excluded.

    Suppose, for a contradiction, that $|A_i| > |A'_i|$ for some agent $i \in N$. Construct the transformation graph $T = \mathcal{G}(\mathcal{A}, \mathcal{A''})$.
    If $T$ contains at least one cycle, pick one of them arbitrarily and remove it---we remove the edges of the cycle from the graph without changing the allocations, so the resulting graph may no longer be the transformation graph of the same allocations.
    Repeat this process until $T$ becomes acyclic.
    Note that this procedure does not change the difference between the indegree and the outdegree of any node.
    
    Since $|A_i| > |A_i'| \ge |A_i''|$, there must be at least one outgoing edge from~$i$. 
    As the graph is finite and has no cycle, by starting from $i$ and following an outgoing edge, we must eventually reach a node~$j$ with no outgoing edge.
    In particular, it holds that $|A_j| < |A''_j|$.
    From $\mathcal{A}$, we can transfer goods along the path from $i$ to~$j$; denote the allocation after this sequence of transfers by $\widehat{\mathcal{A}}$. 
    Since $\mathcal{A}$ is an \mwnwtie{} allocation, it is preferred to $\widehat{\mathcal{A}}$ by the \mwnwtie{} rule.  
    Note that $|\widehat{A}_i| = |A_i| - 1$, $|\widehat{A}_j| = |A_j| + 1$, and $|\widehat{A}_k| = |A_k|$ for all $k\in N\setminus~\{i,j\}$.
    Consequently, the weighted utility vector $(|A_i|^{w_i}, |A_j|^{w_j})$ is preferred to the weighted utility vector $((|A_i|-1)^{w_i}, (|A_j|+1)^{w_j})$ by the \mwnwtie{} rule.\footnote{Even though we write agent~$i$'s weighted utility before agent~$j$'s in the weighted utility vectors, the actual order would be reversed if $j < i$.} 

    Similarly, from $\mathcal{A''}$, we can transfer goods along the path from $j$ to $i$ in the ``reverse'' graph $\mathcal{G}(\mathcal{A''},\mathcal{A})$. 
    Since the allocation $\mathcal{A''}$ is simply the allocation~$\mathcal{A'}$ without the good $g^*$, the same sequence of transfers is also feasible in $\mathcal{A'}$.
    Let the allocation after this sequence of transfers be $\widehat{\mathcal{A}}'$. 
    Since $\mathcal{A'}$ is an \mwnwtie{} allocation, it is preferred to $\widehat{\mathcal{A}}'$ by the \mwnwtie{} rule.
    Note that $|\widehat{A}'_i| = |A'_i| + 1$, $|\widehat{A}'_j| = |A'_j| - 1$, and $|\widehat{A}'_k| = |A'_k|$ for all $k\in N\setminus\{i,j\}$.
    Hence, the weighted utility vector $(|A'_i|^{w_i},|A'_j|^{w_j})$ is preferred to the weighted utility vector $((|A'_i|+1)^{w_i},(|A'_j|-~1)^{w_j})$ by the \mwnwtie{} rule. 
    
    Now, recall that we have 
    \begin{equation*}
        |A_i| > |A'_i| \geq |A''_i| \quad \text{ and } \quad |A'_j| \geq |A''_j| > |A_j| .
    \end{equation*}
    Since bundle sizes are integers, we get
    \begin{equation}\label{inq-resourcemonotone-2}
        |A_i| - 1 \geq |A'_i| \geq |A_i''| \quad\text{ and }\quad |A'_j| \geq |A''_j| \ge |A_j|+1.
    \end{equation}
    In particular, $|A_i| \ge 1$ and $|A'_j| \ge |A''_j| \geq 1$.
    We consider the following two cases based on whether $A_j$ is empty.
    \begin{description}
        \item[Case 1: $|A_j| = 0$.] If $|A_i| > 1$, then the allocation $\widehat{\mathcal{A}}$, with weighted utility vector $((|A_i|-1)^{w_i}, (|A_j|+1)^{w_j})$ for agents $i$ and $j$, has strictly more agents receiving positive utility, thereby contradicting the fact that $\mathcal{A}$ is an \mwnwtie{} allocation. Thus, $|A_i| = 1$, which by (\ref{inq-resourcemonotone-2}) means that $|A_i'| = 0$.
        In particular, $g^*\not\in A_i'$.
        Since $(|A_i|^{w_i}, |A_j|^{w_j})$ is preferred to $((|A_i|-1)^{w_i}, (|A_j|+1)^{w_j})$, it must be that $i < j$.
        
        We know from (\ref{inq-resourcemonotone-2}) that $|A_j'| \ge 1$.
        If $|A'_j| > 1$, then the allocation $\widehat{\mathcal{A}}'$, with weighted utility vector $((|A'_i|+1)^{w_i},(|A'_j|-1)^{w_j})$ for agents $i$ and $j$, has strictly more agents receiving positive utility, thereby contradicting the fact that $\mathcal{A}'$ is an \mwnwtie{} allocation. 
        Thus, $|A'_j| = 1$. 
        However, since the weighted utility vector $(|A'_i|^{w_i},|A'_j|^{w_j})$ is preferred to $((|A'_i|+1)^{w_i},(|A'_j|-~1)^{w_j})$, we must have $j < i$, a contradiction with the previous paragraph.
        
        \vspace{3mm}
        
        \item[Case 2: $|A_j| > 0$.] Then, by (\ref{inq-resourcemonotone-2}), we have  $|A'_j| \ge |A''_j| > 1$. Since $(|A_i|^{w_i}, |A_j|^{w_j})$ is preferred to $((|A_i|-1)^{w_i}, (|A_j|+1)^{w_j})$, we have
        \begin{equation}\label{eqn-resourcemonotone-case4-1}
            \frac{(|A_i|-1)^{w_i}(|A_j|+1)^{w_j}}{|A_i|^{w_i}|A_j|^{w_j}} \leq 1,
        \end{equation}
        where the denominator is nonzero as $|A_i| \geq 1$ and $|A_j| > 0$.
        From (\ref{inq-resourcemonotone-2}), we have $|A_i| \geq |A_i'| + 1$ and $|A'_j| - 1 \geq |A_j|$, which give us 
        \begin{equation}\label{eqn-resourcemonotone-case4-2}
            -\frac{1}{|A'_i|+1} \leq -\frac{1}{|A_i|} \quad\text{ and }\quad \frac{1}{|A'_j|-1} \leq \frac{1}{|A_j|}.
        \end{equation}
        All the denominators in the above inequalities are nonzero as $|A_i| \geq 1$, $|A_j| > 0$, and $|A'_j| > 1$.
        Also, since $(|A'_i|^{w_i}, |A'_j|^{w_j})$ is preferred to $((|A'_i|+~1)^{w_i}, (|A'_j|-1)^{w_j})$, we get
        \begin{equation}\label{eqn-resourcemonotone-case4-3}
            \frac{|A'_i|^{w_i}|A'_j|^{w_j}}{(|A'_i|+1)^{w_i}(|A'_j|-1)^{w_j}} \geq 1,
        \end{equation}
        where the denominator is nonzero as $|A'_j| > 1$.
        Now,
        \begin{align}\label{eqn-resourcemonotone-case4-4}
            \frac{|A'_i|^{w_i}|A'_j|^{w_j}}{(|A'_i|+1)^{w_i}(|A'_j|-1)^{w_j}} & = \left( 1 - \frac{1}{|A'_i|+1} \right)^{w_i} \left( 1 + \frac{1}{|A'_j|-1} \right)^{w_j} \nonumber \\
            & \leq \left( 1 - \frac{1}{|A_i|} \right)^{w_i} \left( 1 + \frac{1}{|A_j|} \right)^{w_j} \nonumber \\
            & = \frac{(|A_i|-1)^{w_i}(|A_j|+1)^{w_j}}{|A_i|^{w_i}|A_j|^{w_j}} \\
            & \leq 1, \nonumber 
    \end{align}
    where the second line is derived using (\ref{eqn-resourcemonotone-case4-2}) and last line is as a result of  (\ref{eqn-resourcemonotone-case4-1}). 
    Since \mwnwtie{} chooses a lexicographically dominating allocation among all MWNW allocations, (\ref{eqn-resourcemonotone-case4-1}) can be an equality only if $i < j$, and (\ref{eqn-resourcemonotone-case4-3}) can be an equality only if $j < i$. Combining (\ref{eqn-resourcemonotone-case4-1}), (\ref{eqn-resourcemonotone-case4-3}), and (\ref{eqn-resourcemonotone-case4-4}) yields
    \begin{equation*}
        1 \leq \frac{|A'_i|^{w_i}|A'_j|^{w_j}}{(|A'_i|+1)^{w_i}(|A'_j|-1)^{w_j}} \leq \frac{(|A_i|-1)^{w_i}(|A_j|+1)^{w_j}}{|A_i|^{w_i}|A_j|^{w_j}} \leq 1,
    \end{equation*}
    where either the leftmost or the rightmost inequality has to be strict, thereby giving us a contradiction.
    \end{description} 
    In both cases, we have arrived at a contradiction, completing the proof. \hfill $\square$ 
    
\end{proof}

The resource-monotonicity of \mwnwtie{} for binary valuations reveals an important structural property of \mwnwtie{} allocations: when a new valued good is added, exactly one agent's bundle will increase in size (although existing goods may get reallocated).

\begin{corollary}\label{cor-monotoneonegood}
    Let $\mathcal{I}$ and $\mathcal{I'}$ be instances such that $\mathcal{I'}$ can be obtained from $\mathcal{I}$ by adding one valued good, and suppose that $\mathcal{A}$ and $\mathcal{A'}$ are their respective \emph{\mwnwtie{}} allocations. Then, for exactly one agent $i \in N$, $v_i(A_i') = v_i(A_i) + 1$, and for all other agents $j \neq i$, $v_j(A_j') = v_j(A_j)$. 
\end{corollary}

\begin{algorithm}[tb]
    \caption{\mwnwtie{} algorithm}
    \label{algo-mwnw}
    \textbf{Input}: Set of agents $N=\{1,\dots,n\}$, set of goods $G=\{g_1,\dots,g_m\}$, weight vector $\mathbf{w} = (w_1, \dots, w_n)$, and valuation vector $\mathbf{v} = (v_1, \dots, v_n)$\\
    \textbf{Output}: \mwnwtie{} allocation of the $m$ goods in $G$
    
    \begin{algorithmic}[1]
        \STATE Initialize the empty allocation $\mathcal{A}^0$ where $A_i^0 = \emptyset$ for all $i \in N$.
        \FOR{$t = 1,2,\dots,m$}
        \STATE $\mathcal{A}^t \leftarrow \texttt{AddOneGood}(\mathcal{A}^{t-1}, g_t)$ (see Algorithm~\ref{alg:addonegood}) \\
        \ENDFOR
        \STATE \textbf{return} $\mathcal{A}^m$
    \end{algorithmic}
\end{algorithm}

\begin{algorithm}[tb]
    \caption{\texttt{AddOneGood}}
    \label{proc-addonegood}
    \textbf{Input}: \mwnwtie{} allocation $\mathcal{A}^{t-1}$ ($t-1$ goods in total), and good $g$\\
    \textbf{Output}: \mwnwtie{} allocation $\mathcal{A}^{t}$ ($t$ goods in total)
    
    \begin{algorithmic}[1]
    \IF{$g$ is unvalued}
        \STATE \textbf{return} the allocation $\mathcal{A}^{t-1}$ with $g$ not assigned to any agent
    \ENDIF
    \STATE Create a directed graph $\mathcal{H}$ where the nodes correspond to the $n$ agents and an edge $x \rightarrow y$ exists whenever agent $y$ values at least one good in agent $x$'s bundle in $\mathcal{A}^{t-1}$. 
    \STATE $\mathcal{P} \leftarrow \emptyset$ 
    \STATE Create a dummy agent represented by node $d$, and bundle $A^{t-1}_d = \{g\}$. For each agent $i \in N$, add an edge $d \rightarrow i$ if $i$ values $g$.
    \FOR{each agent $i = 1, \dots, n$}
        \IF{a path $P_{d,i}$ from $d$ to $i$ exists}
            \STATE Add $(i, P_{d,i}, \mathbf{u}_{d,i})$ to $\mathcal{P}$, where $\mathbf{u}_{d,i}$ is the weighted utility vector of the allocation $\mathcal{A}^{t-1}$ if we were to add one good valued by $i$ to $i$'s bundle. 
        \ENDIF
    \ENDFOR
    \STATE Select the tuple in $\mathcal{P}$ with the maximum weighted Nash welfare across all tuples (computed using $\mathbf{u}_{d,i}$), breaking ties according to Definition~\ref{def:mwnwtie}. Let the corresponding path be $d \rightarrow a_1 \rightarrow a_2 \rightarrow \dots \rightarrow i$. Allocate $g$ to $a_1$, some good in $A_{a_1}^{t-1}$ that $a_2$ values to $a_2$, some good in $A_{a_2}^{t-1}$ that $a_3$ values to $a_3$, and so on. Let this new allocation be $\mathcal{A}^t$. \label{line:select}
    \STATE \textbf{return} $\mathcal{A}^t$
    \end{algorithmic}
\label{alg:addonegood}
\end{algorithm}

We leverage this property to devise an algorithm that computes an \mwnwtie{} allocation in polynomial time (Algorithm~\ref{algo-mwnw}).
By abstracting out the subroutine \texttt{AddOneGood} (Algorithm~\ref{alg:addonegood}), the main algorithm starts with a trivial \mwnwtie{} allocation (with no good being allocated), and iteratively allocates one good at a time using \texttt{AddOneGood}. 
At the end of each call to \texttt{AddOneGood}, the partial allocation will be \mwnwtie{} with respect to the allocated goods.
The correctness of the algorithm is shown in the following theorem.

\begin{theorem}\label{thm-mwnwalgo}
    Under binary valuations, Algorithm~\ref{algo-mwnw} computes an \emph{\mwnwtie{}} allocation in polynomial time.
\end{theorem}

\begin{proof}
    To prove that the allocation returned by Algorithm~\ref{algo-mwnw} is an \mwnwtie{} allocation, we show that after every iteration, the allocation returned by the subroutine \texttt{AddOneGood} is \mwnwtie{} with respect to the goods that are already allocated. 
    
    Let $t \in \{1,\dots,m\}$ be some iteration of the \textbf{for} loop of Algorithm~\ref{algo-mwnw}, and let $g$ be the good added during this iteration.
    If $g$ is unvalued, clearly the allocation $\mathcal{A}^{t-1}$ with $g$ not assigned to any agent is \mwnwtie{} with respect to the first $t$ goods.
    Assume therefore that $g$ is valued.
    Hence, in the graph $\mathcal{H}$, there is an edge $d\rightarrow i$ for some $i\in N$.
    This means that at least one tuple is added to $\mathcal{P}$, and this iteration terminates with an allocation $\mathcal{A}^t$.
    
    Let $\widehat{\mathcal{A}}^t$ be an \mwnwtie{} allocation of the $t$ goods.
    By Corollary~\ref{cor-monotoneonegood}, exactly one agent $z$ has $v_z(\widehat{A}^t_z) = v_z(A_z^{t-1}) + 1$, while all other agents $y\ne z$ have $v_y(\widehat{A}^t_y) = v_y(A_y^{t-1})$.
    Let $\widehat{\mathcal{A}}^{t-1}$ be the allocation obtained by adding a dummy agent $d$ and the good $g$ to $\mathcal{A}^{t-1}$, so that $\widehat{A}_d^{t-1} = \{g\}$.
    Consider the transformation graph $\mathcal{G}(\widehat{\mathcal{A}}^{t-1}, \widehat{\mathcal{A}}^t)$, where we assume that $\widehat{A}_d^t = \emptyset$.
    As in the proof of Theorem~\ref{thm-resourcemonotone}, we can eliminate all cycles from this graph.
    Node $d$ has one outgoing edge and no incoming edge, while every other node except $z$ has the same number of outgoing and incoming edges (in particular, if such a node has an incoming edge, it must also have an outgoing edge).
    Since the graph is finite and has no cycle, by starting from $d$ and iteratively following an outgoing edge, we must eventually reach $z$.
    In other words, there is a path from $d$ to $z$ in $\mathcal{G}(\widehat{\mathcal{A}}^{t-1}, \widehat{\mathcal{A}}^t)$.
    Let this path be $d\rightarrow b_1\rightarrow b_2\rightarrow\dots\rightarrow z$.
    
    By construction of the path, agent~$b_1$ values some good in agent~$d$'s bundle, agent~$b_2$ values some good in agent~$b_1$'s bundle, and so on.
    This means that the path $d\rightarrow b_1\rightarrow b_2\rightarrow\dots\rightarrow z$ also exists in the graph $\mathcal{H}$ in Algorithm~\ref{alg:addonegood}.
    Hence, a tuple $(z,P_{d,z},\mathbf{u}_{d,z})$ is added to $\mathcal{P}$, where $P_{d,z}$ is a path from $d$ to $z$ (which may be different from the aforementioned path).
    Since the weighted utility vector $\mathbf{u}_{d,z}$ is the same as the weighted utility vector of $\widehat{\mathcal{A}}^t$, and by the selection method in line~\ref{line:select} of Algorithm~\ref{alg:addonegood}, the allocation $\mathcal{A}^t$ returned by \texttt{AddOneGood} is also an \mwnwtie{} allocation.
    This completes the proof of correctness.
    
    Finally, we analyze the time complexity of Algorithm~\ref{algo-mwnw}.
    In the main algorithm, there are $\mathcal{O}(m)$ iterations of the \textbf{for} loop. In the \texttt{AddOneGood} subroutine, checking whether $g$ is unvalued takes $\mathcal{O}(n)$ time, constructing the graph $\mathcal{H}$ takes $\mathcal{O}(mn)$ time, and
    there are $\mathcal{O}(n)$ iterations of the \textbf{for} loop. Within each iteration, finding if a path exists can be done using breadth-first search or depth-first search, which takes $\mathcal{O}(n^2)$ time~\cite{cormen2009introduction}. 
    In the last step, selecting the tuple in $\mathcal{P}$ takes $\mathcal{O}(n)$ time and allocating along the path also takes $\mathcal{O}(n)$ time. Putting everything together, Algorithm~\ref{algo-mwnw} terminates in $\mathcal{O}(mn(m+n^2))$ time. \hfill $\square$  
\end{proof}

Next, with the help of the subroutine \texttt{AddOneGood}, we establish the population-monotonicity of \mwnwtie{}.

\begin{theorem}\label{thm-populationmonotone}
    Under binary valuations, the \emph{\mwnwtie{}} rule is population-monotone.
\end{theorem}
\begin{proof}
    Consider an \mwnwtie{} allocation $\mathcal{A}$ for $n+1$ agents. Let $S$ be any subset of $n$ agents, and let $i$ be the agent not in $S$. 
    By Lemma~\ref{lemma-subset-mwnw}, $\mathcal{A}_S$ is an \mwnwtie{} allocation for the agents in $S$. 
    Then, with the bundle of goods remaining to be allocated being that of agent $i$, iteratively make use of the subroutine \texttt{AddOneGood} to allocate all of these goods, giving us an \mwnwtie{} allocation of the original set of goods. 
    Since \texttt{AddOneGood} never decreases the size of any agent's bundle, which is equal to the agent's utility by Lemma~\ref{lemma-mwnw-value}, population-monotonicity is satisfied. \hfill $\square$
\end{proof}

\section{Group-strategyproofness}\label{sec:strategyproofness}
A central concept in mechanism design is the notion of \textit{strategyproofness} (SP), which is a compelling guarantee to prevent strategic manipulation by agents. 
Strategyproof rules ensure that an agent cannot find himself better off should he lie about his valuation function. 
A more robust version of SP is known as \emph{group-strategyproofness} (GSP), wherein no group of agents can misreport their preferences so as to obtain a strictly higher utility for all agents in that group.

\begin{definition}[Group-strategyproofness]
    An allocation rule $f$ satisfies \emph{group-strategyproofness (GSP)} if there do not exist valuation profiles $\mathbf{v}$ and $\mathbf{v'}$ and a nonempty group of agents $S \subseteq N$ such that 
    \begin{itemize}
        \item $v_{i}(A'_{i}) > v_{i}(A_{i}) \text{ for all } i \in S $, and
        \item $v_j = v'_j$ for all $j \in N \setminus S$,
    \end{itemize}
     where $\mathcal{A} = f(\mathbf{v})$ and $\mathcal{A'} = f(\mathbf{v'})$. 
\end{definition}

Note that GSP reduces to SP if we only consider singleton groups.
Halpern et al.~\cite{halpern2020binary} showed that in the unweighted setting, \mnwtie{} satisfies GSP. We strengthen their result by extending it to the weighted setting.

\begin{theorem}
\label{thm-GSP}
    Under binary valuations, the \emph{\mwnwtie{}} rule is  group-strategyproof.
\end{theorem}

\begin{proof}
    Let $\mathcal{A}^\texttt{truth}$  be an \mwnwtie{} allocation under a truthful valuation profile~$\mathbf{v}$, with the unallocated (and truthfully unvalued) goods put into $A^\texttt{truth}_\text{unvalued}$.
    
    Similarly, let $\mathcal{A}^\texttt{lie}$ be an \mwnwtie{} allocation under the same truthful valuations for all agents except possibly those in $S \subseteq N$---call this valuation profile~$\mathbf{v}'$---with the unallocated goods put into $A^\texttt{lie}_\text{unvalued}$. Note that $A^\texttt{lie}_\text{unvalued}$ contains goods that agents in $N \setminus S$ truthfully do not value, and that agents in $S$ declare they do not value (even though some of them may actually value some goods here).

        Suppose, for a contradiction, that there exist such valuation profiles $\mathbf{v}$ and~$\mathbf{v'}$ and a set of agents $S \subseteq N$ such that
    \begin{enumerate}[(i)]
        \item $v_{i}(A_{i}^\texttt{lie}) > v_{i}(A_{i}^\texttt{truth}) \text{ for all } i \in S $,
        \item $v_j = v'_j$ for all $j \in N \setminus S$.
    \end{enumerate}      
    Under truthful valuations, Lemma~\ref{lemma-mwnw-value} tells us that every agent values every good in his own bundle. However, because agents in $S$ may not be truthful, they may receive goods they do not actually value. For every $i \in S$, we can write
    \begin{equation*}
        A_i^\texttt{lie} = A_i^\texttt{lie,valued} \cup  A_i^\texttt{lie,unvalued},
    \end{equation*}
    where $A_i^\texttt{lie,valued}$ consists of the goods in $A_i^\texttt{lie}$ that $i$ values, and $A_i^\texttt{lie,unvalued}$ consists of the goods in $A_i^\texttt{lie}$ that $i$ does not value (but declares as valued in~$\mathbf{v}'$). So we have that
    \begin{equation*}
        v_i(A_i^\texttt{lie}) = v_i(A_i^\texttt{lie,valued}).
    \end{equation*}
    Combining this with condition (i), we get that, for all $i \in S$,
    \begin{equation}\label{eqn-lievaluedtruth}
        |A_i^\texttt{lie,valued}| = v_i(A_i^\texttt{lie,valued}) = v_i(A_i^\texttt{lie}) > v_i(A_i^\texttt{truth}) = |A_i^\texttt{truth}|,
    \end{equation}
    where the last equality follows from Lemma \ref{lemma-mwnw-value}. Let $\mathcal{A}^\texttt{*lie}$ be the allocation that results from excluding the set of goods $\bigcup_{i \in S} A_i^\texttt{lie,unvalued}$ from $\mathcal{A}^\texttt{lie}$. 
    Also, let $A_\text{unvalued}^\texttt{*lie} = A_\text{unvalued}^\texttt{lie}$.
    
    Construct the graph $T = \mathcal{G}(\mathcal{A}^\texttt{truth} \cup \{A_\text{unvalued}^\texttt{truth}\},  \mathcal{A}^\texttt{*lie} \cup \{A_\text{unvalued}^\texttt{*lie}\})$, where we use the notation $\mathcal{A}^\texttt{truth} \cup \{A_\text{unvalued}^\texttt{truth}\}$ to mean the allocation $\mathcal{A}^\texttt{truth}$ with the bundle $A_\text{unvalued}^\texttt{truth}$ added to it, and similarly for $\mathcal{A}^\texttt{*lie} \cup \{A_\text{unvalued}^\texttt{*lie}\}$.
    Note that there are $n+1$ nodes in $T$: each of the first $n$ nodes represents an agent in $N$, while the last node represents the dummy agent $\textit{unvalued}$. We make the following claim:
    \begin{description}
        \item[Claim] There are no edges emerging from the dummy agent \textit{unvalued}.
    \end{description}

     To see why the Claim holds, recall that none of the $n$ agents truthfully values any good in $A_\text{unvalued}^\texttt{truth}$. This means that any edge emerging from the node \textit{unvalued} can only be as a result of some agent $i \in S$ misreporting that he values it (when in fact he does not). However, such a good would go into $A_i^\texttt{lie,unvalued}$, which is explicitly excluded in our transformation graph $T$.
     
     The rest of the proof proceeds using similar ideas as the proof of Theorem~\ref{thm-resourcemonotone}. First, we can eliminate all cycles from $T$.
     Let $\ell\in S$.
     Since $|A_\ell^\texttt{lie,valued}| > |A_\ell^\texttt{truth}|$ by (\ref{eqn-lievaluedtruth}), there must be an incoming edge to node~$\ell$, and this edge cannot come from the node \emph{unvalued} by our Claim.
     Since the graph is finite and all cycles have been eliminated, by starting from $\ell$ and iteratively following an incoming edge, we must eventually reach a node~$j$ with no incoming edge.
     In particular, it holds that $|A_j^\texttt{*lie}| < |A_j^\texttt{truth}|$. 
     By condition (i), we have $j\not\in S$, and 
     by the Claim, no node on the path from $j$ to $\ell$ is the node \emph{unvalued}.
     
     Traverse the path from $j$ to $\ell$, and let $i$ be the first agent belonging to $S$. Then, on the path from $j$ to $i$, all agents (except $i$) belong to $N \setminus S$. 
     From $\mathcal{A}^\texttt{truth}$, we can transfer goods along the path from $j$ to $i$; let the allocation after this sequence of transfers be $\widehat{\mathcal{A}}^\texttt{truth}$. 
     Since $\mathcal{A}^\texttt{truth}$ is an \mwnwtie{} allocation, it is preferred to  $\widehat{\mathcal{A}}^\texttt{truth}$ by the \mwnwtie{} rule. 
     Note that $|\widehat{A}^\texttt{truth}_i| = |A^\texttt{truth}_i| + 1$, $|\widehat{A}^\texttt{truth}_j| = |A^\texttt{truth}_j| - 1$, and $|\widehat{A}^\texttt{truth}_k| = |A^\texttt{truth}_k|$ for all $k\in N\setminus\{i,j\}$, and
     observe that every transferred good in $T$ is transferred from an agent who truthfully values it to another agent who truthfully values it.
     Consequently, the weighted utility vector $(|A_i^\texttt{truth}|^{w_i},|A_j^\texttt{truth}|^{w_j})$ is preferred to the weighted utility vector $((|A_i^\texttt{truth}|+1)^{w_i},(|A_j^\texttt{truth}|-1)^{w_j})$ by the \mwnwtie{} rule.\footnote{Even though we write agent~$i$'s weighted utility before agent~$j$'s in the weighted utility vectors, the actual order would be reversed if $j < i$.} 

     Similarly, from $\mathcal{A}^\texttt{*lie}$, we can transfer a good along the reverse path from $i$ to $j$. Observe that in this transfer, every agent values (under $\mathbf{v'}$) the good that he receives, because for every agent $h$ receiving a good along this path, $h  \in N \setminus S$, which implies $v_h = v'_h$ by condition (ii). As a result, this path from $i$ to $j$ also exists when going from $\mathcal{A}^\texttt{lie}$ to $\mathcal{A}^\texttt{truth}$.
     Let the allocation after such a transfer be $\widehat{\mathcal{A}}^\texttt{lie}$. Since $\mathcal{A}^\texttt{lie}$ is an \mwnwtie{} allocation, it is preferred to $\widehat{\mathcal{A}}^\texttt{lie}$ by the \mwnwtie{} rule according to the valuation profile $\mathbf{v'}$. 
     Note that $|\widehat{A}^\texttt{lie}_i| = |A^\texttt{lie}_i| -~1$, $|\widehat{A}^\texttt{lie}_j| = |A^\texttt{lie}_j| + 1$, and $|\widehat{A}^\texttt{lie}_k| = |A^\texttt{lie}_k|$ for all $k\in N\setminus\{i,j\}$.
     Hence, the weighted utility vector $(|A_i^\texttt{lie}|^{w_i},|A_j^\texttt{lie}|^{w_j})$ is preferred to the weighted utility vector $((|A_i^\texttt{lie}|-1)^{w_i},(|A_j^\texttt{lie}|+1)^{w_j})$ by the \mwnwtie{} rule.
     
     Now, recall that we have
     \begin{equation*}
        |A_i^\texttt{lie}| \geq |A_i^\texttt{*lie}| > |A_i^\texttt{truth}| 
        \quad \text{ and } \quad
        |A_j^\texttt{truth}| > |A_j^\texttt{*lie}| = |A_j^\texttt{lie}|,
     \end{equation*}
     where the last equality holds since $j\not\in S$. 
     Since bundle sizes are integers, we get
     \begin{equation}\label{inq-strategyproof-1}
        |A_i^\texttt{lie}| \geq |A_i^\texttt{*lie}| \ge  |A_i^\texttt{truth}|+1 
        \quad \text{ and } \quad
        |A_j^\texttt{truth}|-1 \ge |A_j^\texttt{*lie}| = |A_j^\texttt{lie}|.
     \end{equation}
     In particular, $|A_i^\texttt{lie}| \geq |A_i^\texttt{*lie}| \ge 1$ and $|A_j^\texttt{truth}| \ge 1$.
     We consider the following two cases based on the size of $|A_j^\texttt{truth}|$.
     \begin{description}
         \item[Case 1: $|A_j^\texttt{truth}| = 1$.] Then $|A_j^\texttt{lie}| = 0$. If $|A_i^\texttt{lie}| > 1$, then the allocation $\widehat{\mathcal{A}}^\texttt{lie}$, with weighted utility vector $((|A_i^\texttt{lie}|-1)^{w_i},(|A_j^\texttt{lie}|+1)^{w_j})$ for agents $i$ and $j$, has strictly more agents receiving positive utility, thereby contradicting the fact that $\mathcal{A}^\texttt{lie}$ is an \mwnwtie{} allocation. Thus, $|A_i^\texttt{lie}| = 1$, which by (\ref{inq-strategyproof-1}) means that $|A_i^\texttt{truth}| = 0$. 
         Since $(|A_i^\texttt{lie}|^{w_i},|A_j^\texttt{lie}|^{w_j})$ is preferred to $((|A_i^\texttt{lie}|-~1)^{w_i},(|A_j^\texttt{lie}|+1)^{w_j})$, it must be that $i < j$. However, since the weighted utility vector $(|A_i^\texttt{truth}|^{w_i},|A_j^\texttt{truth}|^{w_j})$ is preferred to $((|A_i^\texttt{truth}|+1)^{w_i},(|A_j^\texttt{truth}|-1)^{w_j})$, we must have $j < i$, a contradiction.
         
         \vspace{3mm}
         
         \item[Case 2: $|A_j^\texttt{truth}| > 1$.] If $|A_i^\texttt{truth}| = 0$, then the allocation $\mathcal{A}^\texttt{truth}$, with weighted utility vector $((|A_i^\texttt{truth}|+1)^{w_i},(|A_j^\texttt{truth}|-1)^{w_j})$ for agents $i$ and $j$, has strictly more agents receiving positive utility, thereby contradicting the fact that $\mathcal{A}^\texttt{truth}$ is an \mwnwtie{} allocation. Thus, $|A_i^\texttt{truth}| >~0$, which by (\ref{inq-strategyproof-1}) means that $|A_i^\texttt{lie}| > 1$. If $|A_j^\texttt{lie}| = 0$, then we arrive at the same contradiction as we did at the beginning of Case 1. Thus, $|A_j^\texttt{lie}| > 0$, which by~(\ref{inq-strategyproof-1}) means that $|A_j^\texttt{truth}| > 1$. Since $(|A_i^\texttt{truth}|^{w_i},|A_j^\texttt{truth}|^{w_j})$ is preferred to  $((|A_i^\texttt{truth}|+1)^{w_i},(|A_j^\texttt{truth}|-1)^{w_j})$, we have
        \begin{equation}\label{eqn-strategyproof-1}
            \frac{|A_i^\texttt{truth}|^{w_i}|A_j^\texttt{truth}|^{w_j}}{(|A_i^\texttt{truth}|+1)^{w_i}(|A_j^\texttt{truth}|-1)^{w_j}} \geq 1,
        \end{equation}
        where the denominator is nonzero as $|A_j^\texttt{truth}| > 1$. 
        From (\ref{inq-strategyproof-1}), we have $|A_i^\texttt{lie}| \ge  |A_i^\texttt{truth}|+1 $ and $|A_j^\texttt{truth}|-1 \ge |A_j^\texttt{lie}|$, which give us
        \begin{equation}
        \label{eqn-strategyproof-temp}
            -\frac{1}{|A_i^\texttt{lie}|} \geq -\frac{1}{|A_i^\texttt{truth}| + 1} \quad \text{ and } \quad \frac{1}{|A_j^\texttt{lie}|} \geq \frac{1}{|A_j^\texttt{truth}| - 1}.
        \end{equation}
        All the denominators in the above inequalities are nonzero as $|A_i^\texttt{lie}| > 1$, $|A_j^\texttt{lie}| > 0$, and $|A_j^\texttt{truth}| > 1$. 
        Also, since $(|A_i^\texttt{lie}|^{w_i},|A_j^\texttt{lie}|^{w_j})$ is preferred to $((|A_i^\texttt{lie}|-1)^{w_i},(|A_j^\texttt{lie}|+1)^{w_j})$, we get 
         \begin{equation}\label{eqn-strategyproof-2}
             \frac{(|A_i^\texttt{lie}|-1)^{w_i}(|A_j^\texttt{lie}|+1)^{w_j}}{|A_i^\texttt{lie}|^{w_i}|A_j^\texttt{lie}|^{w_j}} \leq 1,
        \end{equation}
        where the denominator is nonzero as $|A_i^\texttt{lie}| > 1$ and $|A_j^\texttt{lie}| > 0$. Now,
        \begin{align} \label{eqn-strategyproof-3}       \frac{(|A_i^\texttt{lie}|-1)^{w_i}(|A_j^\texttt{lie}|+1)^{w_j}}{|A_i^\texttt{lie}|^{w_i}|A_j^\texttt{lie}|^{w_j}} & = \left( 1 - \frac{1}{|A_i^\texttt{lie}|} \right)^{w_i} \left( 1 + \frac{1}{|A_j^\texttt{lie}|} \right)^{w_j} \nonumber \\
                & \geq \left( 1 - \frac{1}{|A_i^\texttt{truth}|+1} \right)^{w_i} \left( 1 + \frac{1}{|A_j^\texttt{truth}|-1} \right)^{w_j} \nonumber \\
                & = \frac{|A_i^\texttt{truth}|^{w_i}|A_j^\texttt{truth}|^{w_j}}{(|A_i^\texttt{truth}|+1)^{w_i}(|A_j^\texttt{truth}|-1)^{w_j}}  \\
                & \geq 1, \nonumber
        \end{align}
        where the second line follows from (\ref{eqn-strategyproof-temp}) and the last line is as a result of (\ref{eqn-strategyproof-1}). 
        Since \mwnwtie{} chooses a lexicographically dominating allocation among all MWNW allocations, (\ref{eqn-strategyproof-1}) can be an equality only if $j < i$, and (\ref{eqn-strategyproof-2}) can be an equality only if $i < j$.
        Combining (\ref{eqn-strategyproof-1}), (\ref{eqn-strategyproof-2}), and (\ref{eqn-strategyproof-3}) yields 
        
        \begin{align*}
            1 &\geq \frac{(|A_i^\texttt{lie}|-1)^{w_i}(|A_j^\texttt{lie}|+1)^{w_j}}{|A_i^\texttt{lie}|^{w_i}|A_j^\texttt{lie}|^{w_j}} \geq \frac{|A_i^\texttt{truth}|^{w_i}|A_j^\texttt{truth}|^{w_j}}{(|A_i^\texttt{truth}|+1)^{w_i}(|A_j^\texttt{truth}|-1)^{w_j}} \geq 1,
        \end{align*}
        where either the leftmost or the rightmost inequality has to be strict, thereby giving us a contradiction.
    \end{description}
    In both cases, we have arrived at a contradiction, completing the proof. \hfill $\square$ 
\end{proof}

An even stronger version of GSP, \textit{strong GSP}, states that no group of agents can misreport their preferences so as to obtain a strictly higher utility for \textit{some} of the agents in that group without hurting the other members of the group.

\begin{definition}[Strong group-strategyproofness]
    An allocation rule $f$ satisfies \emph{strong group-strategyproofness (strong GSP)} if there do not exist valuation profiles $\mathbf{v}$ and $\mathbf{v'}$ and a nonempty group of agents $S \subseteq N$ such that for some nonempty $C\subseteq S$,
    \begin{itemize}
        \item $v_{i}(A'_{i}) > v_{i}(A_{i}) \text{ for all } i \in C $,
        \item $v_{k}(A'_{k}) = v_{k}(A_{k}) \text{ for all } k \in S \setminus C$, and
        \item $v_j = v'_j$ for all $j \in N \setminus S$,
    \end{itemize}
    where $\mathcal{A} = f(\mathbf{v})$ and $\mathcal{A'} = f(\mathbf{v'})$. 
\end{definition}

If the first condition is altered to hold for all $i \in S$ and the second condition is removed, we obtain the standard version of GSP as introduced previously. Like GSP, strong GSP reduces to SP if we only consider singleton groups.
We complement Theorem~\ref{thm-GSP} by showing that \mwnwtie{} fails strong GSP even in the unweighted setting (in which case it reduces to \mnwtie{}).

\begin{proposition}
\label{prop:strong-GSP}
    In the unweighted setting, \emph{\mnwtie{}} does not satisfy strong group-strategyproofness even under binary valuations.
\end{proposition}

\begin{proof}
    Consider an unweighted instance with three agents $N=\{1,2,3\}$ and four goods $G=\{g_1,g_2,g_3,g_4\}$, with agents' valuation profile $\mathbf{v}$:

    \begin{center}
        \begin{tabular}{ c | c c c c }
         $\mathbf{v}$ & $g_1$ & $g_2$ & $g_3$ & $g_4$ \\ 
         \hline
         $1$ & \circled{1} & \circled{1} & 0 & 0 \\  
         $2$ & 0 & 1 & \circled{1} & 0 \\
         $3$ & 0 & 0 & 1 & \circled{1} \\
        \end{tabular}
    \end{center}
    
    The \mnwtie{} rule returns the circled allocation $A_1 = \{g_1, g_2\}$, $A_2 = \{g_3\}$, $A_3 = \{g_4\}$ (recall that the rule break ties lexicographically), and the bundle values are $v_1(A_1) = 2$ and $v_2(A_2) = v_3(A_3) = 1$.
    Now, suppose agents $2$ and $3$ form a manipulating coalition and misreport their values for goods $g_3$ and $g_2$, respectively. Denote this new valuation profile by $\mathbf{v'}$: 
    
    \begin{center}
        \begin{tabular}{ c | c c c c }
         $\mathbf{v'}$ & $g_1$ & $g_2$ & $g_3$ & $g_4$ \\ 
         \hline
         $1$ & \circled{1} & 1 & 0 & 0 \\  
         $2$ & 0 & \circled{1} & \textbf{0} & 0 \\
         $3$ & 0 & \textbf{1} & \circled{1} & \circled{1} \\
        \end{tabular}
    \end{center}
    Now, the \mnwtie{} rule returns the circled allocation $A_1 = \{g_1\}, A_2 = \{g_2\}, A_3 = \{g_3,g_4\}$, and the bundle values are $v_1(A_1) = v_2(A_2) = 1$ and $v_3(A_3) = 2$.  In this case, the misreporting of preferences by the group of agents $\{2,3\}$ led to a strictly higher utility for agent~$3$ and the same utility for agent~$2$, hence violating strong group-strategyproofness.\footnote{Note that in this example, even if agent~$3$ does not lie (i.e., $v'_3 = v_3$), we arrive at the same outcome.}
    $\hfill \square$
\end{proof}

\section{Conclusion and Future Work}\label{sec:conclusion}

In this work, we show that the maximum weighted Nash welfare rule with lexicographic tie-breaking, \mwnwtie{}, exhibits several desirable properties when agents have binary valuations.
In particular, the rule satisfies resource- and population-monotonicity as well as group-strategyproofness, and can be implemented in polynomial time.
Together with previous results on its fairness and efficiency \cite{chakraborty2020wef}, our results thus indicate that \mwnwtie{} is an attractive rule in the binary setting with arbitrary entitlements.

While prior work and ours have demonstrated a number of features of \mwnwtie{}, it is still an open question whether \mwnwtie{} (or its variants based on different tie-breaking of MWNW) is the unique rule that offers these features.
An axiomatic characterization would further strengthen the case for \mwnwtie{}---such a characterization is not known even for \mnwtie{} in the unweighted setting as far as we are aware.
Nevertheless, we briefly discuss how other candidate rules fall short in comparison to \mwnwtie{}.
\begin{itemize}
\item One possible rule is a version of \emph{serial dictatorship}, where the first agent picks all goods that she values, then the second agent, and so on. While it is not difficult to check that this rule is resource- and population-monotone, group-strategyproof, and Pareto-optimal, it can produce an extremely unfair outcome.
Indeed, if all agents value every good and have equal weights, then the rule allocates all goods to the first agent, even though it is clear that a fair allocation should spread out the goods equally in this case.
\item In order to make the picks more balanced, a prominent class of rule to use is the class of \emph{picking sequences} \cite{chakraborty2020wef,chakraborty2021picking}.
In the unweighted setting, a well-known picking sequence is the \emph{round-robin algorithm}, wherein the agents take turns picking their favorite goods in the order $1,2,\dots,n,1,2,\dots,n,1,2,\dots$ until the goods run out (we assume that agents break ties in favor of lower-indexed goods).
However, a round-robin allocation is not necessarily Pareto-optimal, as seen in the following example:
    \begin{center}
        \begin{tabular}{ c | c c }
         $\mathbf{v}$ & $g_1$ & $g_2$ \\ 
         \hline
         $1$ & \circled{1} & 1 \\  
         $2$ & 1 & \circled{0} \\
        \end{tabular}
    \end{center}
The returned (circled) allocation $A_1 = \{g_1\}, A_2 = \{g_2\}$ is Pareto-dominated by the allocation $A_1 = \{g_2\}, A_2 = \{g_1\}$.

Moreover, the round-robin algorithm fails (individual) strategyproofness.
To see this, consider the following instance:
    \begin{center}
        \begin{tabular}{ c | c c c c c c }
         $\mathbf{v}$ & $g_1$ & $g_2$ & $g_3$ & $g_4$ & $g_5$ & $g_6$ \\
         \hline
         $1$ & \circled{1} & \circled{1} & 1 & 1 & \circled{0} & 0 \\  
         $2$ & 0 & 0 & \circled{1} & \circled{1} & 1 & \circled{1} \\
        \end{tabular}
    \end{center}
According to the algorithm, agent~$1$ picks $g_1$, agent~$2$ picks $g_3$, agent~$1$ picks $g_2$, agent~$2$ picks $g_4$, agent~$1$ picks $g_5$, and agent~$2$ picks $g_6$, resulting in the circled allocation.
Yet, agent~$1$ can benefit from the following manipulation:
    \begin{center}
        \begin{tabular}{ c | c c c c c c }
         $\mathbf{v'}$ & $g_1$ & $g_2$ & $g_3$ & $g_4$ & $g_5$ & $g_6$ \\
         \hline
         $1$ & \circled{1} & \circled{\textbf{0}} & 1 & \circled{1} & 0 & 0 \\  
         $2$ & 0 & 0 & \circled{1} & 1 & \circled{1} & \circled{1} \\
        \end{tabular}
    \end{center}
This time, agent~$1$ picks $g_1$, agent~$2$ picks $g_3$, agent~$1$ picks $g_4$, agent~$2$ picks $g_5$, agent~$1$ picks $g_2$, and agent~$2$ picks $g_6$.
As a result, agent~$1$'s utility with respect to her true valuation function increases from $2$ to $3$.
\item Instead of letting the agents pick, one could take a global view and try to maximize a certain welfare notion, as MWNW does.
In the unweighted setting, a commonly used notion is \emph{utilitarian welfare}, i.e., the sum of the agents' utilities.
However, with binary valuations, optimizing the utilitarian welfare can result in varying degrees of fairness.
For instance, when every good is valued by all agents, every allocation maximizes the utilitarian welfare (which is always $m$, the number of goods), but an allocation that gives all goods to the same agent is clearly unfair.
\item A reasonable welfare notion for achieving fairness is \emph{egalitarian welfare}.
Since there are often several allocations that maximize the egalitarian welfare, a typical refinement is the \emph{leximin} rule \cite{moulin2003fairdivision}, which maximizes the smallest utility, then the second smallest utility, and so on.
Halpern et al.~\cite{halpern2020binary} showed that with equal weights, the set of leximin allocations coincides with that of MNW allocations, so the known properties of MNW also apply to leximin (assuming that ties are broken according to \mnwtie{}).
A natural extension of leximin to the weighted setting is to optimize the weighted utilities\footnote{For example, if all $w_i$'s are integers, there are $w_1 + \dots + w_n$ goods, and each agent has value $1$ for each good, then like MWNW, this extension ensures that agent~$i$ receives $w_i$ goods, which is intuitively the fair allocation for this instance.} $v_i(A_i)/w_i$. 
However, as Chakraborty et al.~\cite{chakraborty2022weighted} noted, the resulting rule exhibits problematic behavior even in the simple case of two agents and one valuable good: since the minimum utility is always $0$, it allocates the good to the agent with a smaller weight, as this makes the ratio $v_i(A_i)/w_i$ larger.
More generally, if there are $n$ agents and $m < n$ goods valued by all agents, weighted leximin allocates one good each to the $m$ agents with the smallest weights.
\end{itemize}

We end this paper with some additional directions for future work.
\begin{itemize}
\item As we discuss in Section~\ref{sec:intro}, the advantages of \mwnwtie{} that we show in this paper cease to exist for agents with general additive valuations, even when the agents have equal entitlements.
It would therefore be interesting to determine whether there are intermediate valuation classes between binary and general additive for which some of these advantages can be recovered.
For example, Akrami et al.~\cite{akrami2022maximizing} recently considered \emph{$2$-valued} instances, wherein there exist positive integers $p \le q$ such that $v_i(g)\in\{p,q\}$ for all $i\in N$ and $g\in G$, and showed that maximizing the Nash welfare (in the unweighted setting) is computationally tractable if $p$ divides $q$ but intractable otherwise.
\item Since \mwnwtie{} fails strong group-strategyproofness even in the unweighted setting, it is natural to ask whether there are rules that simultaneously satisfy strong group-strategyproofness and fairness/monotonicity properties.
Note that satisfying strong group-strategyproofness alone is trivial, for example, by ignoring the agents' preferences and choosing a fixed allocation.
\end{itemize}

\subsubsection*{Acknowledgments}
This work was partially supported by an NUS Start-up Grant.
We thank the anonymous reviewers for their valuable comments.

\bibliographystyle{splncs04}
\bibliography{abb,nic}

\end{document}